\documentclass[a4paper]{article}

\title{Minimal TSP Tour is coNP--Complete}
\author{Marzio De Biasi \\
\texttt{\footnotesize{marziodebiasi [at] gmail [dot] com}}
}

\date{} 

\usepackage{verbatim} 
\usepackage[pdftex]{graphicx}
\usepackage{url}
\usepackage{amsthm}
\usepackage[english]{babel}
\usepackage[utf8x]{inputenc}
\usepackage{amsmath}
\usepackage{graphicx}
\usepackage[colorinlistoftodos]{todonotes}
\usepackage{hyperref}

\usepackage{needspace}

\usepackage{color}

\definecolor{clgray}{rgb}{0.25,0.25,0.25}

\theoremstyle{plain}
\newtheorem{theorem}{Theorem}[section]
\newtheorem{corollary}[theorem]{Corollary}

\theoremstyle{definition}
\newtheorem{definition}[theorem]{Definition}

\theoremstyle{remark}

\newcommand{\TSPDECISION}{{\sc TSPDecision}}
\newcommand{\TSPMINDECISION}{{\sc TSPMinDecision}}
\newcommand{\NPC}{{\sf NP}--complete}
\newcommand{\NP}{{\sf NP}}
\newcommand{\FPNP}{$\sf{FP^{NP}}$}
\newcommand{\CONP}{{\sf coNP}}
\newcommand{\TSPEXACT}{{\sc TSPExact}}
\newcommand{\TSPCOST}{{\sc TSPCost}}
\newcommand{\TSP}{TSP}
\newcommand{\DP}{{\sf DP}}
\newcommand{\SAT}{{\sc SAT}}
\newcommand{\THREESAT}{{\sc 3SAT}}
\newcommand{\THREECNF}{{\sc 3CNF}}
\newcommand{\FOURCNF}{{\sc 4CNF}}
\newcommand{\TSPANOTHERTOUR}{{\sc TSPAnotherTour}}

\newenvironment{myquote}{\list{}{\leftmargin=0.2in\rightmargin=0.3in}\item[]}{\endlist}

\begin{document}

\maketitle

\begin{abstract}
The problem of deciding if a Traveling Salesman Problem (TSP)
tour is minimal was proved to be \CONP{}--complete by 
Papadimitriou and Steiglitz. We give an alternative proof
based on a polynomial time reduction from \THREESAT{}.
Like the original proof, our reduction also shows that
given a graph $G$ and an Hamiltonian path of $G$,
it is \NPC{} to check if $G$ contains an Hamiltonian cycle (Restricted
Hamiltonian Cycle problem).
\end{abstract}

\section{Introduction}

The \emph{Traveling Salesman Problem} (\TSP{}) is a well--known problem from
graph theory \cite{PapadimitriouComplexity},\cite{GJ}: we are given $n$
cities and a nonnegative integer distance $d_{ij}$ between 
any two cities $i$ and $j$ (assume that the distances are symmetric, i.e.
for all $i,j, d_{ij} = d_{ji}$). We are asked to find the \emph{shortest tour}
of the cities, that is a permutation $\pi$ of $[1..n]$ such that 
$\sum_{i=1}^n d_{\pi(i),\pi(i+1)}$ (where $\pi(n+1) = \pi(n)$) is as small as possible. Its decision version is the following: 

\begin{myquote}
\TSPDECISION: If a nonnegative
integer bound $B$ (the traveling salesman's ``budget'') is given along
with the distances, does there exist
a tour of all the cities having total length no more than $B$?
\end{myquote}

\noindent\TSPDECISION{} is \NPC{} (we assume that the reader is familiar with
the theory of \NP--completeness, for a good introduction see \cite{GJ} or \cite{Sipser}).
In \cite{PapadimitriouComplexity} two other problems are introduced:

\begin{myquote}
\TSPEXACT: Given the distances $d_{ij}$ among the $n$ cities and an nonnegative integer $B$, is the length of
the shortest tour \emph{equal} to $B$; and
\\
\\
\TSPCOST: Given the distances $d_{ij}$ among the $n$ cities calculate the \emph{length} of the shortest tour.
\end{myquote}

\noindent\TSPEXACT{} is \DP--complete (a language $L$ is in the class \DP{} if and only if there are two languages $L_1 \in \NP$
and $L_2 \in \CONP$ and $L = L_1 \cap L_2$); \TSPCOST{} and \TSP{} are both \FPNP--complete 
(\FPNP{} is the class of all functions from strings to strings that can be computed by
a polynomial--time Turing machine with a \SAT{} oracle) \cite{PapadimitriouComplexity}.

Recently a post by Jean Francois Puget: 
``\href{https://www.ibm.com/developerworks/community/blogs/jfp/entry/no_the_tsp_isn_t_np_complete}{No, The TSP Isn't NP Complete}'' and the subsequent reply by Lance Fortnow: 
``\href{http://blog.computationalcomplexity.org/2014/01/is-traveling-salesman-np-complete.html}{Is Traveling Salesman NP-Complete?}'' \cite{blog:tspmindecision} (re--)raised the question of the correct interpretation
of the statement ``TSP is \NPC{}'';
indeed, if we are given a tour, checking that it is the shortest tour seems
not to be in \NP{}.
A question about the complexity of the following problem:

\begin{myquote}
\TSPMINDECISION{}: Given  a set of $n$ cities, the distance between all city pairs and a tour $T$, is T visiting each city exactly once and is T of minimal length?
\end{myquote}

\noindent was posted on \mbox{cstheory.stackexchange.com}, a question and answer site for professional researchers in theoretical computer science and related fields
\cite{cstheory:tspmindecision}. 

We gave an answer with a first sketch of the proof that \TSPMINDECISION{} is \CONP{}--complete,
but after formalising and publishing it on arXiv, \href{http://cstheory.stackexchange.com/a/21644/3247}{we discovered that the result
is not new} and it originally appeared in \cite{papasteitsp} (see also Section~19.9 in \cite{combopt}). The proof given by Papadimitriou and Steiglitz is different:
they prove that the Restricted Hamiltonian Cycle (RHC) problem is \NPC{} 
starting from an instance of the Hamiltonian cycle problem $G$ and modifying $G$
into a new graph $G'$ that contains an Hamiltonian path, and has an Hamiltonian
cycle if and only if the original $G$ has an Hamiltonian cycle.
Our alternative proof is a chain of reductions from \THREESAT{} to 
the problem of finding a tour shorter than a given one,
and it may be interesting in and of itself, so we decided not to
withdraw the paper.


\section{Minimal TSP tour is coNP--complete}
\label{sec:proof}

Proving that \TSPMINDECISION{} is \CONP{}--complete is equivalent to proving the
\NP{}--completeness of the following:

\begin{definition}[\TSPANOTHERTOUR{} problem]~\\
\noindent{
\textbf{Instance}: A complete graph $G = (V, E)$ with positive
integer distances $d_{ij}$ between the nodes,
and a simple cycle $C$ that visits all the nodes of $G$.
}
\smallskip

\noindent{\textbf{Question}: Is there a simple cycle $D$ that visits all the nodes of $G$
such the total length of the tour $D$ in $G$ is strictly less than
the total of the tour $C$ in $G$?
}
\end{definition}

\begin{theorem}\label{thm:tspanothertour}
\TSPANOTHERTOUR{} is \NPC{}.
\end{theorem}

\begin{proof}
It is easy to see that a valid solution to the problem can be verified in polynomial
time: just check if the tour $D$ visits all the cities and if its length is strictly
less than the length of the given tour $C$, so the problem is in \NP{}.
To prove its hardness we give a polynomial time reduction from \THREESAT{};
given a \THREECNF{} formula $\varphi$ with $n$ variables $x_1,\ldots,x_n$ and
$m$ clauses $C_1,...,C_m$; we introduce a new dummy variable $z$ and
add it to every clause: $(x_{i_1} \lor x_{i_2} \lor x_{i_3} \lor z)$. We obtain a \FOURCNF{} formula $\varphi^z$ that
has at least one satisfying assignment (just set $u=true$). Note that
every satisfying assignment of $\varphi^z$ in which $z =false$ is also
a satisfying assignment of $\varphi$.

From $\varphi^z$ we generate an undirected graph $G = \{V,E\}$ following the same standard transformation used
to prove that the Hamiltonian cycle problem is \NPC{}: for every
clause we add a node $c_j$, for every
variable $x_i$ we add a \emph{diamond--like} component, 
and we add a directed edge from one of the nodes of the diamond
to the node $c_j$ if $x_i$ appears in $C_j$ as a positive literal;
a directed edge from $c_j$ to one of the nodes of the diamond
if $x_i$ appears in $C_j$ as a negative literal.
Starting from the top we can choose to traverse the diamonds
corresponding to variables $x_1,x_2,...,x_n,u$ from left to right
(i.e. set $x_i$ to $true$) or from right to left (i.e. set $x_i$ to $false$).
The resulting directed graph $G$ has an Hamiltonian cycle if and only if the original formula
is satisfiable. For the details of the construction see \cite{Sipser} or \cite{AroraBarak}.

We focus on the diamond corresponding to the dummy variable $z$; let $e_z$
be the edge that must be traversed if we assign to $u$ the value of $true$
(see Figure~\ref{fig:reduction}).

\begin{figure}[htp]
\centering
\includegraphics[width=7cm]{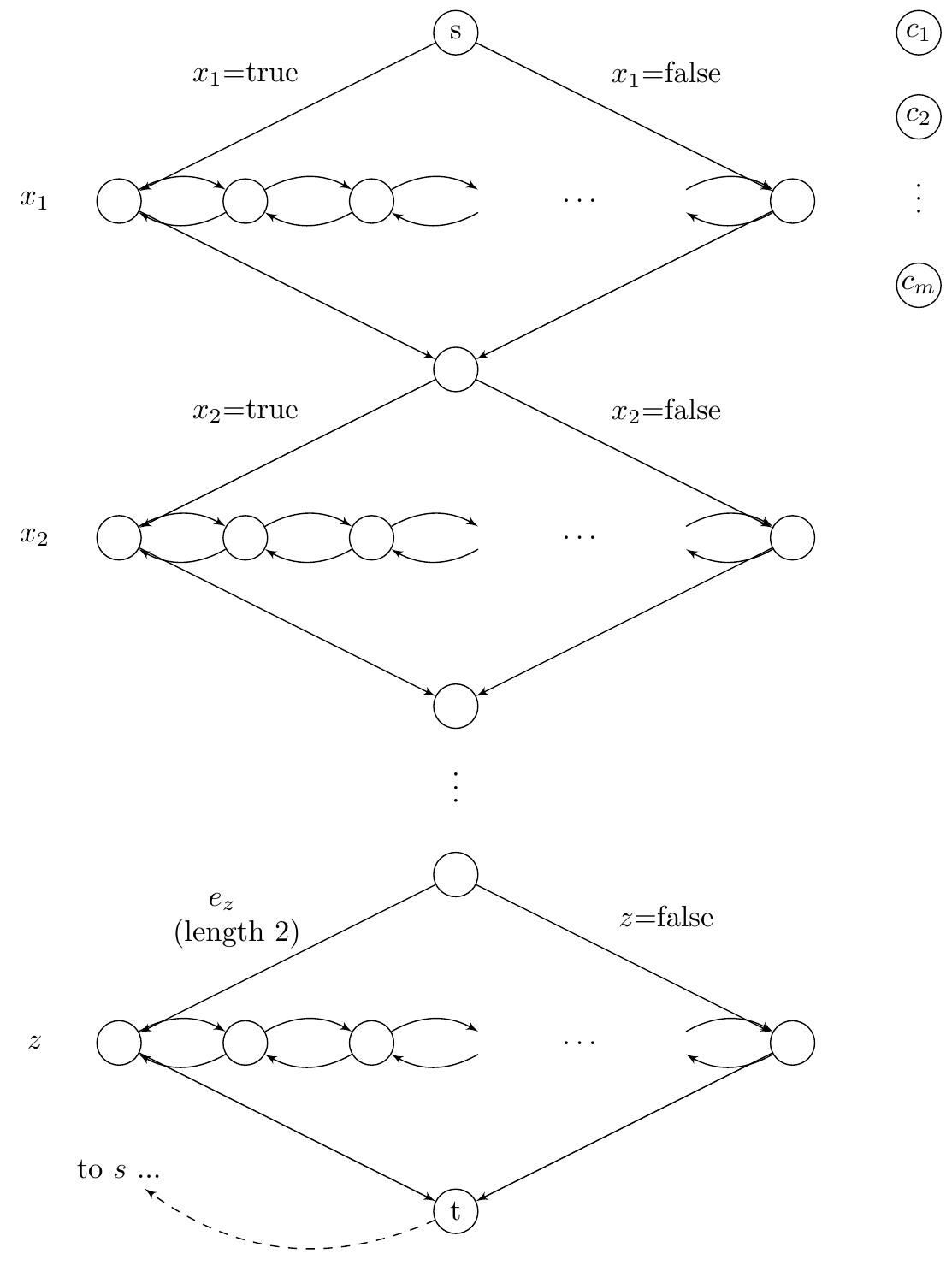}
\caption{Reduction from \THREESAT{} to directed Hamiltonian cycle.}\label{fig:reduction}
\end{figure}

We can transform $G$ to an undirected graph $G' = \{V',E'\}$
replacing each node $u \in V$ with three linked nodes $u_1, u_2, u_3 \in V'$  and modify the edges according to the standard reduction used to prove the \NP{}-completeness of  UNDIRECTED HAMILTONIAN CYCLE from DIRECTED HAMILTONIAN CYCLE \cite{Sipser}:
we use $u_1$ for the incoming edges of $u$, and $u_3$ for the outgoing edges,
i.e. we replace every directed edge $(u \to v) \in E$ with $(u_3 \to v_1) \in E'$.
We have $G'$ has an Hamiltonian cycle if and only if $G$ has an Hamiltonian cycle
if and only if $\varphi^z$ is satisfiable.

Finally we transform $G'$ into an instance of \TSPANOTHERTOUR{} assigning length $1$ to
all edges except edge $e_z$ which has length $2$; and we complete the graph
adding the missing edges and setting their length to $3$.

The dummy variable $z$ guarantees that we can easily find a tour $T$: just 
travel the diamonds from left to right without worrying of the clause nodes;
when we reach the diamond corresponding to $z$, traverse it from left to right
(i.e. assign to $z$ the value of $true$), and include all the $c_j$s. By construction the total length of the tour $T$ is exactly $|V'|+1$: all edges have length 1 except $e_u$ that has length $2$.

Another tour $D$ can have a length strictly less than $|V'|+1$ only if 
it doesn't use the edge $e_u$; so if it exists we can derive a valid
satisfying assignment for the original formula $\varphi$, indeed 
by construction $\varphi$
is satisfiable if and only if there exists a satisfying assignment for
$\varphi^z$ in which $z=false$. In the opposite direction
if there exists a valid satisfying assignment for $\varphi$ we can
easily find a tour $D$ of length $|V'|$: just traverse the diamonds
according to the truth values of the variables $x_i$ and traverse
the diamond corresponding to $z$ from right to left.

So there is another tour $D$ of total length strictly less than $T$ if
and only if the original \THREESAT{} formula $\varphi$ is satisfiable.

\end{proof}

Hence we have:

\begin{corollary}
\TSPMINDECISION{} is \CONP--complete.
\end{corollary}

The reduction used to prove Theorem~\ref{thm:tspanothertour} ``embeds'' the $\sf{NP}$--completeness proof of the \emph{Restricted Hamiltonian Cycle problem}
(RHC) \cite{combopt}:

\begin{theorem}
\label{cor:ham}
Given a graph $G$ and an Hamiltonian path in it, it is \NPC{} to decide
if $G$ contains an Hamiltonian cycle as well.
\end{theorem}

\begin{proof}
In the reduction above, after the creation of
the undirected graph $G'$, if we remove the edge $e_z$,
we are sure that an Hamiltonian path exists from one endpoint
of $e_z$ to the other (just delete $e_z$ from the Hamiltonian cycle that can be constructed setting $z = true$). An Hamiltonian cycle in $E \setminus \{e_z\}$ \emph{must} use
the edge corresponding to $z = false$, so it exists if and only if 
the original \THREESAT{} formula $\varphi$ is satisfiable.

\end{proof}

\section{Conclusion}

We are optimist: if someone -- out there -- shouts: ``TSP is NP--complete''
we are confident that he really means: ``The decision version of TSP is NP--complete'';
and we hope that, soon or later, someone -- out there -- will shout
``We already know that there is [not] a polynomial time algorithm that solves TSP
because $\sf{P}$ is [not] equal to $\sf{NP}$'' :-)

\section*{Acknowledgements}
Thanks to P\'alv\"olgyi D\"om\"ot\"or for the nice hint about 
Theorem~\ref{cor:ham}, and to Marcus Ritt for pointing out
the original Papadimitriou and Steiglitz's paper.

\bibliographystyle{plain} 
\bibliography{tspmindecision} 

\end{document}